\documentclass[11pt,,english]{smfart}
\usepackage{amsmath,amssymb,amsthm,latexsym}
\usepackage{graphicx}
\usepackage{stmaryrd}		
\usepackage{xcolor}
\usepackage{hyperref}
\hypersetup{colorlinks=true,    linkcolor=blue,
citecolor=red,       filecolor=BrickRed,       urlcolor=magenta
}
\usepackage{mathtools} 
\usepackage{enumerate}

\usepackage[capitalize]{cleveref}

\graphicspath{{figures/}}

\newcounter{alphathm}

\newtheorem{theorem}{Theorem}[section]

\newtheorem{lemma}[theorem]{Lemma}
\newtheorem{proposition}[theorem]{Proposition}

\theoremstyle{definition}
\newtheorem{definition}[theorem]{Definition}

\newtheorem{remark}[theorem]{Remark}
\newtheorem{theorem*}[alphathm]{Theorem}
\numberwithin{equation}{section} 
\newcommand{\NN}{\mathbb N}
\newcommand{\ZZ}{\mathbb Z}

\newcommand{\RR}{\mathbb R}

\newcommand{\FF}{\mathbb F}
\newcommand{\Fp}{\FF_p}

\newcommand{\bR}{R}
\newcommand{\F}{{\mathsf F}}

\newcommand{\colonequal}{\mathrel{\mathop:}=}

\newcommand{\tcm}{}
\newcommand{\bk}{k}
\newcommand{\bh}{\boldsymbol h}
\newcommand{\bt}{\boldsymbol t}
\newcommand{\br}{\boldsymbol r}
\newcommand{\bi}{\boldsymbol i}

\newcommand{\mB}{\mathcal B}
\newcommand{\res}{{\mathsf{res}}}

\newcommand{\Card}{{\rm Card}}
\newcommand{\NP}{{\rm NP}}
\newcommand{\GL}{{\rm GL}}

\usepackage{float}
\usepackage[noend]{algpseudocode}
\usepackage[shortlabels,inline]{enumitem}
\newfloat{algo}{tb}{lop}
\floatname{algo}{Algorithm}

\def\pFqnoargs#1#2{{}_#1F_#2}
\def\pFq#1#2#3#4#5#6{\pFqnoargs{#1}{#2}\biggl(\begin{matrix}%
{#3}\kern.707em{#4}\\{#5}%
\end{matrix}\,\bigg|\,#6\biggr)}   

\title[Diagonals and algebraicity modulo $p$]{Diagonals and algebraicity modulo $p$:\\
a sharper degree bound }
\date{}

\author{Boris Adamczewski}
\address{Univ. Claude Bernard Lyon 1, CNRS UMR 5208, Institut Camille Jordan, 43 blvd. du 11 novembre 1918, F-69622 Villeurbanne Cedex, France \\   and 
Centre International de Rencontres Mathématiques (CIRM),  
163 Av. de Luminy,
13009 Marseille, 
France}
\email{boris.adamczewski@math.cnrs.fr}

\author{Alin Bostan}
\address{Inria, Sorbonne Université, LIP6, 4 place Jussieu, 75252 Paris, France}
\email{alin.bostan@inria.fr}

\author{Xavier Caruso}
\address{CNRS, IMB, Universit\'e de Bordeaux, 351 cours de la Lib\'eration, 33405 Talence, France}
\email{xavier@caruso.ovh}


\thanks{Partially supported by the French grant DeRerumNatura (ANR-19-CE40-0018), and by
the French--Austrian project EAGLES (ANR-22-CE91-0007 \& FWF I6130-N)}

\begin{document}
\maketitle

\begin{altabstract} 
En 1984, Deligne a montr\'e que pour tout nombre premier~$p$, la r\'eduction modulo $p$ de la diagonale d'une s\'erie formelle alg\'ebrique de plusieurs variables \`a  coefficients entiers est alg\'ebrique sur le corps des fonctions rationnelles \`a coefficients dans $\mathbb{F}_p$. De plus, il a sugg\'er\'e que les degr\'es d'alg\'ebricit\'e  $d_p$ de ces fonctions devaient cro\^itre au plus polynomialement en fonction de $p$. Dans cet article, nous pr\'esentons une nouvelle preuve du th\'eor\`eme de Deligne qui est  \'el\'ementaire et permet d'\'etablir la premi\`ere borne g\'en\'erale polynomiale avec un degr\'e raisonnable.  
\end{altabstract}

\begin{abstract}
In 1984, Deligne proved that for any prime number $p$, the reduction modulo $p$ of the diagonal of a multivariate algebraic power series
with integer coefficients is algebraic over the field of rational functions with coefficients in $\mathbb F_p$. Moreover, he conjectured that the
algebraic degrees $d_p$ of these functions should grow at most polynomially in $p$. 
In this article, we provide a new and elementary proof of Deligne's theorem, which yields the first general polynomial bound on $d_p$ with an explicit and reasonable degree.  
\end{abstract}

\setcounter{tocdepth}{1}

\section{Introduction}\label{sec: intro}

Given a ring $R$ and a multivariate power series  
\[
f(\bt)=  \sum_{\bi \in\mathbb{N}^n} a(\bi)\bt^{\bi} \in R[\![\bt]\!]\, ,
\]
where  we write $\bi=(i_1,\ldots,i_n)$, $a(\bi)=a(i_1,\ldots,i_n)$,  $\bt=(t_1,\ldots,t_n)$, and $\bt^{\bi}=t_1^{i_1}\cdots t_n^{i_n}$, 
the \emph{diagonal} of $f$ is the univariate power series  
\[
\Delta(f)(t) \coloneqq  \sum_{i=0}^{+\infty} a(i,\ldots,i) t^i \in R[\![t]\!]\, .
\]  
When $\mathfrak p$ is an ideal of $R$, the \emph{reduction modulo $\mathfrak p$} of $f$ is given by  
\[
f_{\vert \mathfrak p}(\bt) \coloneqq  \sum_{\bi \in\mathbb{N}^n} (a(\bi) \bmod \mathfrak p)\bt^{\bi}\in  (R/\mathfrak p)[\![\bt]\!]\, .
\]  

A particularly interesting case arises when  $R=\overline{\mathbb{Q}}$, the field of algebraic numbers, and the power series $f$ is algebraic over $\overline{\mathbb{Q}}(\bt)$. In this setting, $\Delta(f)$ satisfies a linear differential equation of Picard-Fuchs type and belongs to the class of Siegel's $G$-functions.  We refer the reader to \cite{AB13, Ch15} and the references therein for further discussion of these connections.  Moreover, such power series frequently appear in enumerative combinatorics \tcm{(see \cite[Chap.\ 6]{Stanley_book}, \cite{BMM10} and \cite[Chap.\ 4]{Melczer21})}. Additionally, diagonalization is closely related to integration (see \cite{De84} and \cite[Sec.\ 3]{Ch15}), and in general, $\Delta(f)$ is transcendental over $\overline{\mathbb{Q}}(t)$ (see, for instance, \cite{AB13}).  

By contrast, when $R=k$ is a field of characteristic $p>0$, and $f$ is a multivariate rational power series, Furstenberg \cite{Fu67} proved the following remarkable result:  the diagonal 
$\Delta(f)$  is always algebraic over $k(t)$. Later, Deligne \cite{De84} provided a geometric proof of Furstenberg's theorem and extended it to the case where $f$ is algebraic. He also noted an intriguing connection between these two seemingly opposite situations through reductions modulo $p$.  
Indeed, the relation $\Delta(f)_{\vert \mathfrak p} = \Delta(f_{\vert \mathfrak p})$ holds in general. It follows from Deligne's theorem that if $f\in \mathbb{Z}[\![\bt]\!]$ is algebraic, then $\Delta(f)_{\vert p}$ remains  algebraic over $\mathbb{F}_p(t)$ for all primes~$p$. This naturally raises the question of how the algebraic degree $d_p$ of  $\Delta(f)_{\vert p}$ evolves as $p$ varies.   
Deligne's proof relies on heavy arithmetic geometry machinery and proceeds inductively on the number~$n$ of variables. In \cite{De84}, he suggested that a direct proof would be more satisfactory and could yield a polynomial bound of the form $d_p = O(p^N)$. %

Deligne's work has been highly influential, inspiring several authors \cite{DL,SW,Ha88,Sa87,Sa86}, who independently provided a direct and elementary proof of his theorem. However,
this proof resulted in a very weak nonpolynomial bound on $d_p$. More recently, Bell and the first author \cite{AB13} established the first general polynomial bound, but their proof relies on an inductive argument, leading to  an excessively large value of $N$ in the worst case. A brief discussion of these results is provided in Section \ref{sec:comp}.

Our main contribution  is a new, direct,  and elementary proof of Deligne's theorem,  which yields the first polynomial bound $d_p<p^N$ with a reasonable value for $N$,  expressed in terms of the complexity of the underlying algebraic function.  

The complexity of an algebraic power series $f\in k[\![\bt]\!]$ is traditionally measured in terms of its \emph{degree} and  \emph{height}. 
Since the ring $k[\bt,y]$ is a unique factorization domain, there exists a polynomial $E(\bt,y)\in k[\bt,y]$ satisfying $E(\bt,f)=0$, which is minimal for divisibility.
Moreover such a polynomial is unique up to multiplication by a nonzero constant in~$k$.
The \emph{degree} of $f$ is defined as the degree in~$y$ of $E(\bt,y)$; equivalently, it is the degree of the field extension $k(\bt)(f)$ of~$k(\bt)$. 
For the height, we consider two natural definitions: the \emph{total height} of $f$ is the total degree in $\bt$ of $E(\bt,y)$ (where the total degree of the monomial $\bt^{\bi}$ is  $i_1+\cdots+i_n$), while its \emph{partial height} is the tuple $\bh =(h_1,\ldots,h_n)$ where, for each $i$, $h_i$ is the degree in $t_i$ of $E(\bt, y)$.

We recall that any element $f$ in the algebraic closure of $\mathbb F_p(t)$  is annihilated by a \emph{linearized polynomial},  {\it i.e.}, a polynomial  $P(X)\in\Fp(t)[X]$ of the form $$c_0X+c_1X^p+\cdots +c_NX^{p^N} \,, \quad\quad c_N\not=0\,.$$ 
The integer $N$ is called the \emph{$p$-degree} of  $P$. 
It easily follows that the Galois conjugates of $f$ are all contained in
an $\Fp$-vector space of dimension $N$. Consequently, the Galois group of $f$
(\emph{i.e.},  the Galois group of the extension of $\Fp(t)$ generated by $f$ and all its
Galois conjugates) canonically embeds, up to conjugacy, into $\GL_N(\Fp)$.

Our main result is stated as follows.

\begin{theorem}
\label{thm: modp}
Let $f \in \ZZ[\![\bt]\!]$ be an algebraic power series with degree $d$, total height $h$, and partial height $(h_1,\ldots,h_n)$.  Set
\begin{equation}
\label{eq:N}
N \coloneqq (d+1)\cdot \min \left \{ 
\prod_{i=1}^n (h_i+1) - 
\prod_{i=1}^n h_i
, 
\binom{n+h}{n}-\binom{h}{n} \right\}\,.
\end{equation}
Then, for every prime number $p$, $\Delta(f)_{|p}$ is annihilated by a linearized polynomial of $p$-degree at most $N$. 
In particular, $\Delta(f)_{|p}$ has degree at most $p^N-1$ over $\mathbb
F_p(t)$.  
\end{theorem}

Let us make a few comments about this result. 

\begin{remark}\label{rem:GLN} 
It follows from Theorem~\ref{thm: modp} and the  preceding remark, that, for every prime \( p \), the Galois group of \( \Delta(f)_{|p} \) embeds, up to conjugacy,  
into \( \mathrm{GL}_N(\mathbb{F}_p) \).  
The key point here is that \( N \) does not depend on \( p \).  
This observation naturally leads to more refined questions concerning uniformity with respect to 
$p$. In particular, one may ask whether the Galois groups of 
$\Delta(f)_{|p}$ arise by reduction modulo
$p$ from a unique group, or a finite number of groups, defined in characteristic zero. 
In a recent preprint \cite{CFV}, Caruso, F\"urnsinn, and Vargas-Montoya investigate this problem.
\end{remark}

\begin{remark}\label{rem: genus}
\tcm{%
We will deduce Theorem~\ref{thm: modp} from a slightly more precise
statement (see \cref{thm:gendiag}) in which the bound $N$ is expressed
in terms of the number of integer points in the Newton polytope of
$E(\bt,y)$.
Beyond the significant improvements this refined version may induce
in some cases, it also suggests a potential link between the optimal
bound $N$ and geometric invariants attached to the algebraic 
hypersurface of equation $E(\bt,y) = 0$.
Indeed, after the work of Baker~\cite{Baker1893}, Hodge~\cite{Hodge29}
and Khovanskii~\cite{Khovanskii78}, it is known that the geometric genus of the
aforementioned hypersurface is bounded from above by the number of integer
points in the Newton polytope of $E(\bt, y)$ and that equality holds
generically.
A natural question is then whether the degree of the minimal polynomial
of $\Delta(f)_{|p}$ is always at most $p^{g + O(1)}$ where $g$ is the
geometric genus of the underlying hypersurface. When $n = 1$, Bridy's approach
\cite{Br17} provides a positive answer to this expectation; in full
generality, however, the question remains open.}
\end{remark}

Let $k$ be a field of characteristic zero, and let $f(\bt)\in k[\![\bt]\!]$ be algebraic over $k(\bt)$. Then there exists a finitely generated $\ZZ$-algebra $Z$ such that 
$f(\bt)\in Z[\![\bt]\!]$, and, for any maximal ideal $\mathfrak p$ of $Z$, the quotient $Z/\mathfrak p$ is a finite field (see, e.g.,~\cite[p.~967]{AB13}).  
In this framework, we obtain, as a direct consequence of Theorem~\ref{thm:diag}, the following generalization of \cref{thm: modp}, in the spirit of  \cite[Thm.\ 1.4]{AB13}.  
A particularly interesting case occurs when \( k \) is a number field and \( Z = \mathcal{O}_{k,\mathfrak{p}} \), the localization of the ring of integers \( \mathcal{O}_k \) of \( k \) at \( \mathfrak{p} \).

\begin{theorem}
\label{thm: modp2}
Let $k$ be a feld of characteristic zero, $Z\subset k$ be a finitely generated $\ZZ$-algebra, and $f \in Z[\![\bt]\!]$ be algebraic over $k(\bt)$ with degree $d$, total height $h$, and partial height $(h_1,\ldots,h_n)$.  
Then, for every maximal ideal $\mathfrak p$ of~$Z$, $\Delta(f)_{|\mathfrak p}$ is annihilated by a linearized polynomial with coefficients in $\mathbb
(Z/\mathfrak p)(t)$ of $p$-degree at most $N$, where $p$ is the characteristic of $Z/\mathfrak p$ and $N$ is defined as in \eqref{eq:N}.  
In particular, $\Delta(f)_{|\mathfrak p}$ has degree at most $p^N-1$ over $\mathbb
(Z/\mathfrak p)(t)$.  
\end{theorem}

\subsection{Comparison with previous works}\label{sec:comp}

The proof of Furstenberg's theorem presented in \cite{Fu67} is direct and elementary, yielding a polynomial upper bound \( d_p < p^N \) when \( d = 1 \) and \( n \) is arbitrary, 
 with  \( N \) having a value similar to that in Theorem~\ref{thm: modp} in this case.  
Deligne~\cite[p.\ 140]{De84} subsequently treated the case \( n = 2 \) and arbitrary \( d \), obtaining the bound $d_p=O(p^N)$, where $N$ is expressed in terms of geometric quantities associated with the underlying algebraic power series. Another approach, which traces its origin to \cite{CKMR},  relies on the fact that any algebraic power series is annihilated by a linearized polynomial. 
This method has been used independently by several authors~\cite{DL,SW,Ha88} (and also~\cite{Sa87,Sa86} in some special cases) to give an elementary proof of Deligne's theorem.
Harase~\cite{Ha89} (see also~\cite{AB12}) later showed that this approach yields a doubly exponential bound, namely 
 \( d_p=O(p^{p^N}) \), for arbitrary $d$ and $n$. 

\tcm{As noted earlier, the first general polynomial bound, {\it i.e.}, in \( O(p^N) \), was established in \cite{AB13}, where the proof provides an effective \( N \) that depends only on the degree \( d \) and the total height \( h \) of~$f$\footnote{Using the approach from \cite{AB13}, it is also possible to deduce the existence of such a polynomial bound from another result in \cite{DL}, but with an ineffective constant \( N \).}.
However, before this paper, the best upper bound for $N$ was a \emph{nonelementary} primitive recursive bound --precisely, a tower of exponentials with height (at least) linear in the number~$n$ of variables.
The reason is that, when \( d > 1 \), the value of \( N \) becomes exceedingly large due to a recursive procedure involving resultants. For instance, even when the number of variables is \( n = 2 \), the estimate for \( N = N(2,d,h) \) in 
\cite[Theorem~6.1]{AB13}
takes the form
\[
2^{2^{2^{4^{hd^2}}}}.
\]
This can be seen by carefully analyzing the quantities appearing in the proof of Theorem~6.1, together with the estimates of Lemmas~6.1, 6.2 and 6.3 in~\cite{AB13}. Indeed, $N(2,d,h)$ grows like $N(1,d_1,h_1) \sim 4 h_1 d_1^3$, where $d_1$ grows like $d_0^{d_0 M^2}$ and $h_1$
grows like $d_0^{d_0 \cdot M^{2d_0^{M^2}}}$, where $M \sim 4hd^3$ and $d_0 \sim d^{4^{hd^2}}$,
and these estimates altogether imply that $\log_2\log_2\log_2\log_2 (h_1)$ grows like $2hd^2$.
}

The approach we adopt here is analogous to that in \cite{AB13}, but instead of expressing algebraic power series as diagonals of rational functions, we express them as formal residues of rational functions. This formulation, which was first employed in \cite{BCCD} for univariate algebraic power series, is a variant of Furstenberg's formula analogous to Lagrange's  formula for the residues of rational functions with simple poles (see Lemma~\ref{lem: furstenberg}). The primary advantage of this method is that it avoids any inductive process.  

Beyond diagonals of algebraic power series, other significant  families of $G$-functions in
 $\mathbb Q[\![t]\!]$ have algebraic reductions modulo $p$ (cf.\ \cite{ABD,VM21}).
Moreover, the property of \emph{algebraicity modulo 
$p$} provides a powerful tool for establishing results on the transcendence and algebraic independence of power series in characteristic zero  (cf.\ \cite{SW89,AGBS,AB13,ABD,VM23}).

\subsection{Organization of the article} 

This article is a condensed version of an unpublished preprint made available on arXiv in 2023 (cf.\ \cite{ABC}). 
In \cref{sec: christol}, we establish \cref{thm: main}, which provides a sharp, quantitative multivariate extension of Christol's theorem (see \cite{Ch79,CKMR}) concerning algebraic power series with coefficients in finite fields. Our result also generalizes its extension to perfect ground fields of positive characteristic, obtained independently in \cite{DL,SW,Ha88}. In Section \ref{sec: diag}, we derive our quantitative version of Deligne's theorem as a consequence of \cref{thm: main}, from which \cref{thm: modp} follows directly. 
We emphasize that \cref{thm: main} is of independent interest and has several further applications,  discussed in detail in \cite{ABC}, including the following:

\begin{itemize}

\item[{\rm (i)}] For a multivariate algebraic power series with coefficients in a finite field~$\mathbb F_q$, it provides 
an upper bound on the minimal number of states required for a $q$-automaton  to
generate its sequence of coefficients. This generalizes a result of Bridy \cite{Br17} to the multidimensional setting (see \cite[Sec.\ 4]{ABC}).

\item[{\rm (ii)}] For two multivariate algebraic power series over an arbitrary field of characteristic $p$, it establishes an upper bound on the algebraic degree of their Hadamard product and other related products. This significantly  improves the doubly exponential bounds that follow from \cite{DL,SW,Ha88} and  were made explicit by Harase \cite{Ha89} (see \cite[Sec.\ ~6]{ABC}). 

\item[{\rm (iii)}] It provides an efficient algorithm for computing the coefficient of a given multivariate algebraic power series  in $\mathbb F_q[\![\bt]\!]$  
at a specified multi-index.  The power series is encoded by 
its minimal polynomial over $\FF_q(\bt)$ along with a sufficient number of initial coefficients to ensure uniqueness,  
Again, this improves significantly upon previously known results (see \cite[Sec.\ 7]{ABC}). 

\end{itemize}

\section{A sharper multivariate Christol theorem}\label{sec: christol}

Let $\bk$ be a perfect field of characteristic $p>0$. 
Then the Frobenius endomorphism~$\F$, which maps $x$ to $x^p$, is an automorphism of $\bk$. 
Let $\bt = (t_1,\ldots,t_n)$ be indeterminates, and define $K_0\coloneqq k(\bt)$, $R \coloneqq k[\![\bt]\!]$ and $K  \coloneqq {\rm Frac}(R)$, the field of fractions of 
$R$.  
The Frobenius map $\F$ extends naturally to $K$ as a field homomorphism by setting $\F(t_i)=t_i^p$, for $1\leq i\leq n$, so that 
for a power series $f \coloneqq\sum_{\bi \in \mathbb N^n}a(\bi)\bt^{\bi}$ $\in
k[\![\bt]\!]$, we have
\[
\F(f) = \sum_{\stackrel{\bi \in \mathbb N^n}{\bi=(i_1,\dots ,i_n)}}a(\bi)^p\bt^{p\bi}\in k[\![\bt]\!] \,.
\]

We let~$K^{\langle p\rangle}$ denote the image of~$K$ by $\F$, so that $\F$ defines an isomorphism
between $K$ and~$K^{\langle p\rangle}$.
Then $K$ is a $K^{\langle p\rangle}$-vector space of dimension $p^n$, a basis being given by all
monomials of the form $\bt^{\br}\coloneqq t_1^{r_1}\cdots t_n^{r_n}$, with $\br
\coloneqq(r_1,\ldots,r_n)\in \{0,\ldots,p-1\}^n$.
Thus, every $f\in K$ has a unique expansion of the form 
\begin{equation}
\label{eq: basis1}
	f= \sum_{\br \in \{0,\ldots,p-1\}^n}\bt^{\br} f_{\br}\,, \quad \mbox{ where }f_{\br}\in K^{\langle p\rangle}.
\end{equation}
\begin{definition}\label{def:secop}
For every $\br\in \{0,\ldots,p-1\}^n$, the \emph{section operator} $S_{\br}$ is the map from $K$ into itself defined by 
\begin{equation}
\label{eq: section1}
S_{\br}(f)\coloneqq \F^{-1}(f_{\br})\,.
\end{equation} 
For a power series $f \coloneqq\sum_{\bi \in \mathbb N^n}a(\bi)\bt^{\bi}$ $\in
k[\![\bt]\!]$, we have
\[
S_{\br}(f) = \sum_{\stackrel{\bi \in \mathbb N^n}{\bi=(i_1,\dots ,i_n)}}a(pi_1+r_1,\ldots,pi_n+r_n)^{1/p}\bt^{\bi}\in k[\![\bt]\!] \,.
\]
We let $\Omega_n$ denote the monoid generated by all section operators under composition.
\end{definition}

The section operators are also sometimes referred to as \emph{Cartier operators} (see, for instance,~\cite{AB13}). 
Recall that they are $k$-linear and satisfy, for all $f,g\in K$, the relation 
\begin{equation}\label{eq:fsecop}
S_{\br}(f\F(g))=S_{\br}(f)g\,.
\end{equation}
Moreover, Equality~\eqref{eq: basis1} can be equivalently written as 
$$
f= \sum_{\br \in \{0,\ldots,p-1\}^n}\bt^{\br} \F(S_{\br}(f))\,,
$$
for all $f\in K$. 

\begin{definition}\label{def: genus} The \emph{Newton polytope} (or,
\emph{Newton polyhedron}) $\NP(A)$ of a multivariate polynomial
 \[A  \coloneqq \sum_{\substack{\bi \in \NN^n \\ j \in \NN}}
  a_{\bi,j} \bt^{\bi} y^j \in \bk[\bt,y]\]
is defined as the convex hull in $\RR^{n+1}$ of the tuples $(\bi,j)$ such that $a_{\bi,j} \neq 0$.
\end{definition}

The main result of this section is stated as follows.   

\begin{theorem} \label{thm: main}
Let $\bk$ be a perfect field of characteristic $p$ and let $\bt=(t_1,\ldots,t_n)$.
Let $A(\bt,y)$ be a nonzero polynomial in $\bk[\bt,y]$, and let $f\in \bk[\![\bt]\!]$ satisfy the
algebraic relation $A(\bt,f)=0$.
Define 
\[ C  \coloneqq \NP(A)+ (-1,0]^{n+1} \,. \]
Then, there exists a $\bk$-vector space $W\subset \bk[\![\bt]\!]$ of dimension at most
\begin{equation*}\label{eq: NP} \Card(C \cap \NN^{n+1})
\end{equation*}
which contains $f$ and  is invariant under the action the monoid $\Omega_n$
of all section operators.
\end{theorem}

\begin{remark}\label{rem: thm}
The plus sign in the definition of $C$ refers to the Minkowski sum. 
In 
\cref{thm: main}, the field $\bk$ must be perfect for the section operators to be well-defined; however, 
this does not present a real limitation. Indeed, replacing an arbitrary field of characteristic $p$ by its
perfect closure does not affect our results (cf.\  Proposition~\ref{prop: descente}). 

In the case where $k$ is a finite field, it follows that the orbit of $f$ under $\Omega_n$ is finite, which implies that the sequence of coefficients of $f$ is 
generated by a finite $p$-automaton.  Consequently, \cref{thm: main} extends and refines Christol's theorem.  
Furthermore, \cref{thm: main} carries a geometric flavor (cf.\  \cref{rem: genus}), analogous to the more recent result established by
Bridy~\cite{Br17} in the case where $k$ is a finite field and $n=1$ (see \cite[Sec.\ 4]{ABC} for
a detailed discussion).
\end{remark}

\subsection{Preliminary results} 

We begin by establishing three auxiliary results: Lemmas \ref{lem: basis} and \ref{lem: furstenberg}, and Proposition \ref{prop: res}. 
These results correspond, respectively, to natural extensions in our framework of Lemma~2.4, Lemma~2.3, and Proposition~2.5 from \cite{BCCD}.

\subsubsection{Section operators on $K(\!(T)\!)$}

Consider a new indeterminate~$T$. Then $\F$ extends to a field homomorphism from
$K(\!(T)\!)$ into itself by setting $\F(T)=T^p$. 
We let $K(\!(T)\!)^{\langle p\rangle}$ denote the image of $K(\!(T)\!)$ under $\F$. As before, $K(\!(T)\!)$ is
a $K(\!(T)\!)^{\langle p\rangle}$-vector space of dimension $p^{n+1}$, with a basis given by all 
monomials of the form $\bt^{\br}T^s$, where $\br \in \{0,\ldots,p-1\}^n$ and $0\leq s\leq p-1$.
However, for our purposes, it will be more convenient to replace this standard basis with a more suitable one, adapted to a given  $f\in R$. 
 
\begin{lemma}
\label{lem: basis}
For any $f\in R$, the family 
\[
\mB_f \colonequal \left(\bt^{\br}(f+T)^s\right)_{(\br,s)\in \{0,\ldots,p-1\}^{n+1}}
\]
is a basis of $K(\!(T)\!)$ as a $K(\!(T)\!)^{\langle p\rangle}$-vector space. 
\end{lemma}

\begin{proof}
First, we observe that $\mB_f$ is a generating family. Indeed, we can obtain $\bt^{\br}T^s$ as a
linear combination of $\bt^{\br}(f+T)^i$, for $0\leq i\leq s$. 
This follows from 
$$
T^s=(T+f-f)^s=\sum_{i=0}^{s} {s \choose i}(-f)^{s-i}(T+f)^{i} \,.
$$
Since $\mB_f$ has the same cardinality as the basis 
$\left(\bt^{\br}T^s\right)_{(\br,s) \in
\{0,\ldots,p{-}1\}^{n+1}}$, 
 it is also a basis of $K(\!(T)\!)$. 
\end{proof}

It follows that, given $f\in R$,  every $x\in K(\!(T)\!)$ has a unique expansion of the form 
\begin{equation}
\label{eq: basis2}
x= \sum_{\br \in \{0,\ldots,p-1\}^n} \bt^{\br} \sum_{s=0}^{p-1} (f+T)^s x_{f,\br,s}\,,  \quad \mbox{ where }x_{f,\br,s} \in K(\!(T)\!)^{\langle p\rangle}.
\end{equation}
For every $\br\in \{0,\ldots,p-1\}^n$ and $s\in\{0,\ldots,p-1\}$, we define the \emph{section 
operator} $S_{f,\br,s}$, from $K(\!(T)\!)$ into itself, by 
\begin{equation}
\label{eq: section2}
S_{f,\br,s}(x)\coloneqq \F^{-1}(x_{f,\br,s})\,.
\end{equation} 
One readily verifies that, for all $x,y\in K(\!(T)\!)$,
$$S_{f,\br,s}(xy^p) = S_{f,\br,s}(x\F(y)) = S_{f,\br,s}(x)y\,,$$ 
 for every $\br\in \{0,\ldots,p-1\}^n$ and  $s\in\{0,\ldots,p-1\}$.  
This definition and the above identity are analogous to Definition~\ref{def:secop} and Equation~\eqref{eq:fsecop}.
\subsubsection{A variant of Furstenberg's formula}

We define the residue map $\res$  from $K(\!(T)\!)$ to $K$ by setting
\[
\res \Bigg(\sum_{n\geq \nu}a_nT^n \Bigg) \coloneqq a_{-1} \,.
\]
Given a polynomial $A \in K[y]$, we let $A_y$ denote its derivative with respect to~$y$. The
following key lemma is  inspired by
\cite[Prop.~2]{Fu67} and analogous to Lagrange's  formula for the residues of rational functions with simple poles.   

\begin{lemma}
\label{lem: furstenberg}
Let $f \in K$ and $A(y)\in K[y]$.
Assume that $A(f)=0$ and $A_y(f)\not=0$.  
Then, for all $P\in K[y]$, one has 
\[
\res \left(\frac{P(f+T)}{A(f+T)}\right)=\frac{P(f)}{A_y(f)} \,\cdot
\] 
\end{lemma}

\begin{proof}
Let $U, V\in K[y]$ such that $V(0)=0$ and $V_y(0)\not=0$. Thus, we can express $V$ as $V=\sum_{n=1}^r a_ny^n$ with $a_1\not=0$. 
It follows that 
\[
\frac{U(T)}{V(T)} = \frac{1}{T}\cdot \left(\frac{U(T)}{a_1+a_2T+\cdots +a_rT^{r-1}}\right) 
\]
and consequently,  
\[
\res \left( \frac{U(T)}{V(T)} \right) = \frac{U(0)}{a_1} = \frac{U(0)}{V_y(0)}\,\cdot
\]
The result follows by applying the above equality to $U=P(f+T)$ and $V=A(f+T)$. 
\end{proof}


\subsubsection{Section operators and residues}

The next result establishes a fundamental commutation relation between taking residues and applying section operators. 
It is reminiscent of a result of Cartier involving the so-called Cartier operator \cite[Th.\ 4]{Car57} (see also \cite[Sec.\ 3]{Br17}).

\begin{proposition}
\label{prop: res}
For any $f\in R$ and $\br \in \{0, \ldots, p{-}1\}^n$, 
the following commutation relation holds over~$K(\!(T)\!)$:
\[S_{\br} \circ \res =  \res \circ S_{f,\br,p-1}\, .\]
\end{proposition}

\begin{proof}
Let $x\in K(\!(T)\!)$.  By Equations~\eqref{eq: basis2} and~\eqref{eq: section2}, we have  
\[
x= \sum_{\br \in \{0,\ldots,p-1\}^n} \bt^{\br} 
	\sum_{s=0}^{p-1} (f+T)^s \, \F(S_{f,\br,s}(x)) \,.
\]
Hence, we obtain that  
\begin{equation}\label{eq: res1} \res (x)  = \sum_{\br \in \{0,\ldots,p-1\}^n}\bt^{\br} \, 
	\sum_{s=0}^{p-1} \res \Big( (f+T)^s \, \F(S_{f,\br,s}(x)) \Big)\,.
\end{equation}
Since $\F(S_{f,\br,s}(x)) \in K(\!(T)\!)^{\langle p\rangle}$, its support (that is, the set of indices corresponding to nonzero coefficients of 
$\F(S_{f,\br,s}(x))$ viewed as a Laurent series in the variable $T$) is contained in $p\mathbb Z$.  Therefore, we have 
\[
\res\Big( (f+T)^s \, \F(S_{f,\br,s}(x)) \Big)=0 \quad \mbox{ for } \quad 0\leq s\leq p-2\,,
\]
while, for $s=p-1$, 
\[
\res\Big( (f+T)^{p-1} \, \F(S_{f,\br,s}(x)) \Big)=\res\Big( T^{p-1} \, \F(S_{f,\br,s}(x)) \Big)=\F\big(\res \left( S_{f,\br,p-1}(x)\big) \right)\,.
\]
Substituting this into Equation~\eqref{eq: res1}, we obtain  
\[		
\res (x)  = \sum_{\br \in \{0,\ldots,p-1\}^n}\bt^{\br} \, \F\big(\res \left( S_{f,\br,p-1}(x)\big) \right) \,.
\]
Finally, applying \eqref{eq: basis1} and \eqref{eq: section1} with $f= \res (x)$, we deduce that 
\[
S_{\br} \circ \res (x)= \res \circ  S_{f,\br,p-1}(x) \,,
\]
as desired. 
\end{proof}

\subsection{Proof of \cref{thm: main}}

Let $E(\bt,y)\in \bk[\bt,y]$ denote the minimal polynomial of $f$
over $k(t)$, normalized such that its coefficients are globally coprime.

\begin{lemma}
\label{lem:separable}
The polynomial $E(\bt,y)$ is separable with respect to $y$. In particular, $f$ is a simple root of $E(\bt, y)$ with respect to $y$.
\end{lemma}

\begin{proof}
Since $E(\bt, y)$ is defined as the minimal polynomial, it suffices to 
prove that $E(\bt, y)$ is not of the form $F(\bt, y^p)$ for some
polynomial $F(\bt,z)\in \bk[\bt,z]$. We proceed by contradiction, assuming that 
$E(\bt, y)=F(\bt, y^p)$ for some  $F(\bt,z)\in \bk[\bt,z]$. We can write 
\[
F(\bt, z) = a_0(\bt) + a_1(\bt) z + \cdots + a_m(\bt) z^m
\]
with $a_i(\bt) \in \bk[\bt]$ and $a_m(\bt) \neq 0$.
Let $\br \in \{0, \ldots, p{-}1\}^n$. Applying the section operator 
$S_{\br}$ to the identity $F(\bt, f^p) = 0$, we obtain
\[
S_{\br}\big(a_0(\bt)\big) + 
S_{\br}\big(a_1(\bt)\big) f + \cdots + 
S_{\br}\big(a_m(\bt)\big) f^m = 0\,.
\]
Since $a_m(\bt)$ is nonzero, there must exist an index $\br$ for
which $S_{\br}(a_m(\bt))\not=0$. 
For such an $\br$, we obtain a nonzero polynomial annihilating
$f$ with $y$-degree smaller than the $y$-degree of~$E$. This
contradicts the minimality of~$E$.
\end{proof}

\begin{proof}[Proof of \cref{thm: main}]
Let $E_y$ denote the partial derivative of
$E$ with respect to $y$. By \cref{lem:separable}, we have 
$E_y(\bt,f)\not=0$.
Furthermore, since $A$ annihilates $f$, it must be a multiple of $E$; that is, we can write $A
= E \cdot F$ for some  $F\in \bk[\bt,y]$.
Let $J$ be the interval $(-1,0]$ and define $C' \colonequal \NP(E) + J^{n+1}$.

We claim that the $\bk$-vector space 
\[
W \coloneqq \left\{ \frac{P(\bt,f)}{E_y(\bt,f)} : 
P\in \bk[\bt,y], \; \NP(P) \subset C' \right\} \subset K 
\]
contains $f$ and is invariant under the action of $\Omega_n$.
The fact that $f \in W$ follows from the observation that $\NP(y E_y) \subset \NP(E) \subset C'$.
Now, consider a tuple $\br \in \{0, 1, \ldots, p{-}1\}^n$ along with a polynomial $P \in
\bk[\bt,y]$ whose Newton polytope is a subset of $C'$.
We define $U \colonequal P \cdot E^{p-1}$, and let $Q \in \bk[\bt,y]$ be defined by 
\begin{equation}
\label{eq:formulaQ}
Q(\bt,f+T)  \coloneqq S_{f,\br,p-1}(U(\bt,f+T))\in K \,.
\end{equation}
By combining~\cref{lem: furstenberg} and \cref{prop: res}, we obtain:
\begin{eqnarray} \label{Sr_P_Q}
S_{\br}\left( \frac{P(\bt,f)}{E_y(\bt,f)}\right)& =& 
	S_{\br} \circ \res  \left( \frac{P(\bt, f+T)}{E(\bt,f+T)}\right)\\\nonumber
&=&  \res  \circ  S_{f,\br,p-1} \left( \frac{P(\bt, f+T)}{E(\bt,f+T)}\right)\\ \nonumber
&=& \res   \left( \frac{Q(\bt, f+T)}{E(\bt,f+T)}\right)\\\nonumber
&=& \frac{Q(\bt,f)}{E_y(\bt,f)} \,\cdot
\end{eqnarray}

To establish our claim, it just remains to prove that $\NP(Q) \subset C'$.
We recall the following standard fact about Newton polytopes: the formation of Newton polytopes is
compatible with products. Specifically,  for $A, B \in \bk[\bt,y]$, we have the relation
\[\NP(AB) = \NP(A) + \NP(B)\,.\]
Using this property, we can derive the following:  
\[\NP(U) \subset (p{-}1){\cdot}\NP(E) + C' = p{\cdot}\NP(E) + J^{n+1}.\]
Now, let $(\bi,j)$ be a tuple of exponents belonging to the support of $Q$, \emph{i.e.}, for which
the coefficient in $Q$ in front of $\bt^{\bi} y^j$ is nonzero.
From the definition of $S_{f,\br,p-1}$, it follows that $(p\bi + \br, pj + p-1)$ must lie in 
$\NP(U)$. 
Dividing by $p$ and defining $I \coloneqq (-\frac 1 p, 0]$, we obtain  
\[\textstyle
\left(\bi + \frac 1 p \br, \, j + \frac{p-1} p\right) \in \NP(E) + 
I^{n+1}\,. \]
Thus, we conclude that 
\[\textstyle
(\bi, j) \in \NP(E) + I^{n+1} + 
\left\{\left(-\frac 1 p \br, -\frac{p-1} p\right)\right\}
\subset \NP(E) + J^{n+1} = C'\,.\]
Finally, we have shown that $\NP(Q) \subset C'$, as desired. 

Clearly $W$ is spanned by the fractions of the form 
$\bt^{\bi} f^j/E_y(\bt,f)$, where $(\bi, j) \in C' \cap \NN^{n+1}$.
Hence, the dimension of $W$ is bounded from above by the cardinality of this set. 
Furthermore, we observe that $C = \NP(F) + C'$, where $F$ is nonzero. Since $\NP(F)$ is the Newton polytope of a nonzero polynomial, it must intersect 
$\NN^{n+1}$. Therefore, $C$ contains a translate of $C'$ by an element with nonnegative integer
coefficients.
As a result, the cardinality of $C\cap \NN^{n+1}$ is at least that of $C'\cap \NN^{n+1}$, and we
conclude that \[\dim_{\bk} W \leq \Card(C'\cap \NN^{n+1})\leq \Card(C\cap \NN^{n+1})\,,\] as desired.
\end{proof}

\section{Diagonals}\label{sec: diag}


By combining \cref{thm: main} with Propositions~5.1 and~5.2 of \cite{AB13}, we immediately  obtain an
effective version of Deligne's theorem: given an algebraic power series $f \in \bk[\![\bt]\!]$ of 
degree $d$ and total height $h$, its diagonal $\Delta(f)$ has degree at most $p^N$ (and height at
most $N p^N$), where $N$ is explicitly given by
\[
  N \coloneqq (d+1)\cdot \binom{n+h}{n}\,.
\]

\subsection{An effective version of Deligne's theorem}
In this section, we establish a refinement of the result stated above, from which \cref{thm: modp} follows directly.

\begin{theorem}
\label{thm:diag}
Let $k$ be an arbitrary field of characteristic $p$. 
Let $f \in \bk[\![\bt]\!]$ be an 
algebraic power series with degree $d$, total height $h$, and partial
height $\bh=(h_1,\ldots,h_n)$. Set
\begin{equation}\label{eq: N31}
N \coloneqq (d+1)\cdot \min \left \{ 
\prod_{i=1}^n (h_i+1) - 
\prod_{i=1}^n h_i
, 
\binom{n+h}{n}-\binom{h}{n} \right\}\,.
\end{equation}
Then, $\Delta(f)$ is annihilated by a linearized polynomial with coefficients in $k(t)$ of $p$-degree at most $N$. 
In particular, $\Delta(f)$ has degree at most $p^N-1$ over $k(t)$. 
\end{theorem}

\cref{thm:diag} will be deduced from \cref{thm:gendiag}, a slightly more general result stated in terms of generalized diagonals and 
Newton polytopes.


\subsection{Generalized diagonals}

In what follows, we introduce a slight generalization of the diagonalization process. We continue with the previous notation: $k$ is a perfect field of characteristic~$p$ and 
$K_0=k(\bt)$, $R= k[\![\bt]\!]$, and  
$K= {\rm Frac}(R)$ are defined as in  \cref{sec: christol}.   
Let $G$ be a subgroup of $\ZZ^n$ such that the quotient $\ZZ^n/G$ has
no torsion. We define $K_{0,G}$ as the subfield of $K_0$ generated by $k$ and
by the monomials $\bt^{\bi}$ with $\bi \in G$. Similarly, we define
$\bR_G$ as the $\bk$-subalgebra of $\bR$ consisting of series of
the form $\sum_{\bi \in G} a(\bi) \bt^{\bi}$. 
Since $G$ is abstractly isomorphic to $\ZZ^m$ for some integer 
$m \leq n$, the rings $K_{0,G}$ and $\bR_G$ are  isomorphic to 
$\bk(x_1,\ldots, x_m)$ and $\bk[\![x_1, \ldots, x_m]\!]$, respectively. 

\begin{definition}
We keep the notation introduced above. 
The \emph{$G$-diagonal} is the operator defined by 
\[\begin{array}{rcl}
\Delta_G : \qquad
\bR & \longrightarrow & \bR_G \medskip \\
\displaystyle \sum_{\bi \in \NN^n} a(\bi) \bt^{\bi}
& \mapsto &
\displaystyle \sum_{\bi \in G} a(\bi) \bt^{\bi}
\end{array}\]
with the convention that $a(\bi) = 0$ for $\bi \not\in \NN^n$.
\end{definition}

When $G$ is the subgroup generated by $(1, \ldots, 1)$, the ring
$R_G$ is isomorphic to $k[\![t]\!]$ \emph{via} the map $t_1 \cdots t_n \mapsto t$, and the diagonal operator $\Delta_G$ 
reduces to the usual diagonal operator $\Delta$.
However, the more general construction of~$\Delta_G$ offers greater flexibility and
allows  for partial diagonals, as considered in \cite{DL}.  For instance, if $G$ is the
subgroup generated by $(1, \ldots, 1)$ together with the standard basis vectors $e_i = (0, \ldots,0, 1, 0, \ldots, 0)$ 
(where $1$ is in $i$-th position) for 
$i \in \{1, \ldots, m\}$, then   $$\bR_G \simeq k[\![t_1, \ldots, t_m, x]\!]$$
and we have 
\[\Delta_G\left(\sum_{i \in \NN^n} a(\bi) \bt^{\bi}\right) \,\, =
\sum_{\substack{(i_1, \ldots, i_m)\in \NN^m \\ k \in \NN}}
a(i_1, \ldots, i_m, k, \ldots, k)\: t_1^{i_1} \cdots t_m^{i_m} x^k\,.\]
More generally, one can verify that $\Delta_G$ is $K_{0,G}$-linear.

\begin{theorem}
\label{thm:gendiag}
Let $k$ be an arbitrary field of characteristic $p$,  and let 
$G$ be a subgroup of $\ZZ^n$ such that $\ZZ^n/G$ has no
torsion.  Let $G_{\RR}$ be the vector subspace of $\RR^n$ 
generated by $G$, and let 
$\pi_G : \RR^{n+1} \to (\RR^n/G_{\RR}) \times \RR$ denote the 
canonical projection. 
Let $A(\bt,y)\in \bk[\bt,y]$ and let $f\in \bk[\![\bt]\!]$ 
satisfy the algebraic relation $A(\bt,f)=0$. Define $C$ as the
convex subset of $\RR^{n+1}$ given by
\[
C  \coloneqq \NP(A) + \big(G_\RR \times (-1,0]\big)\, .
\]
Then,  $\Delta_G(f)$ is annihilated by a linearized polynomial with coefficients in  $K_{0,G}$ of $p$-degree at most $N$, 
where $N  \coloneqq \Card \big(\pi_G(C \cap \NN^{n+1})\big)$.
\end{theorem}

Recall that if $k$ is a field of characteristic~$p$, then adjoining to~$k$ all the
$p^r$-th roots ($r \geq 1$) of all the elements of~$k$ yields a perfect field, called the
perfect closure of $k$, which we denote  by~$k_p$. 
Before proving Theorem~\ref{thm:gendiag}, we recall the following elementary result (see \cite[Prop.\ 3.1]{ABC} for a proof). 

\begin{proposition}\label{prop: descente}
Let $k$ be an arbitrary field of characteristic $p$ and let $k_p$ be its perfect closure.
Let $f\in k[[\bt]]$ be algebraic over $k(\bt)$. 
Then, $[k(\bt)(f):k(\bt)]=[k_{p}(\bt)(f):k_{p}(\bt)]$. 
\end{proposition}

\begin{proof}[Proof of Theorem \ref{thm:gendiag}] 
By Proposition~\ref{prop: descente}, we may, without any loss of
generality, replace the field $k$ by the perfect closure of the subfield of $k$ generated over~$\mathbb F_p$
by the coefficients of $f$. Hence we may assume that $k$ is perfect.

Let $E \in \bk(\bt, y)$ be the minimal polynomial of $f$, and let $E_y$  denote its derivative  with
respect to $y$.
Define $J  \coloneqq (-1,0]$ and 
\[
C'  \coloneqq \NP(E) + \big(G_\RR \times J\big)\, .
\]
By following the proof of \cref{thm: main}, we  obtain that the $\bk$-vector space
\[
W \coloneqq \left\{ \frac{P(\bt,f)}{E_y(\bt,f)} : 
P\in \bk[\bt,y],\, \NP(P) \subset C' \right\}
\]
contains $f$ and is invariant under~$S_{\br}$ for all $\br \in G$. Noticing that $\Delta_G$
commutes with $S_{\br}$ whenever $\br \in G$, it follows that $\Delta_G(W)$ is also invariant under
$S_{\br}$ for all $\br \in G$.

Let $V$ denote the $K_{0,G}$-span of $\Delta_G(W)$ in $K_G \coloneqq \text{Frac}(R_G)$. We first show that the dimension of $V$, viewed as a  $K_{0,G}$-vector space, is bounded by $N$, and then prove that $V$ is invariant under the action of the Frobenius map.  

By linearity, we see that $V$ is spanned by the elements $\frac{\bt^{\bi} f^j}{E_y(\bt,f)}$ for $(\bi, j)$
running over $C' \cap \NN^{n+1}$.
Moreover, two fractions $$\frac{\bt^{\bi} f^j}{E_y(\bt,f)} \quad \mbox{ and }\quad  
\frac{\bt^{\bi'} f^{j'}}{E_y(\bt,f)}$$ are $K_{0,G}$-collinear
when $\bi \equiv \bi' \bmod G$, which occurs if and
only if $\pi_G(\bi, j) = \pi_G(\bi', j)$. 
The dimension of $V$ over $K_{0,G}$ is therefore bounded above by the cardinality of $\pi_G(C' \cap
\NN^{n+1})$, which is itself  bounded by $N$ (see the final paragraph of the proof of~\cref{thm:
main} for more details).

Let us now show that $V$ is invariant under the Frobenius map $\F$. 
The latter acts as an endomorphism of $K_{0,G}$. 
We now define the ``relative'' Frobenius map on $K_G$ by 
\[\begin{array}{rcl}
\psi: \quad 
K_G \otimes_{K_{0,G},\F} K_{0,G} & \longrightarrow & K_G \\
x \otimes y & \mapsto & x^p y
\end{array}\]
where the notation $\otimes_{K_{0,G}, \F}$ indicates that we view
$K_{0,G}$ as an algebra over itself via $\F$.  Hence,
in $K_G \otimes_{K_{0,G}, \F} K_{0,G}$, we have 
$1 \otimes y = y^p \otimes 1$.
This construction ensures that $\psi$ is a $K_{0,G}$-linear
isomorphism. 
Furthermore, $\psi$ is related to the section operators \emph{via} the
formula 
\[\psi^{-1}(f) = \sum_{\br \in G_p} S_{\br}(f) \otimes \bt^{\br}\,,\]
where $G_p \subset G$ is a set of representatives of $G/pG$.
As shown  earlier, $V$ is closed under the action of $S_{\br}$ for all $\br
\in G$. Therefore, we conclude that $\psi^{-1}$ induces a $K_{0,G}$-linear morphism from $V$ to $V
\otimes_{K_{0,G}, \F} K_{0,G}$.
Since  $\psi^{-1}$ is the restriction of an injective map, it is clearly injective.
Furthermore,  because $V$ is finite dimensional over $K_{0,G}$ and  $\dim_{K_{0,G}} V = \dim_{K_{0,G}} (V
\otimes_{K_{0,G}, \F} K_{0,G})$, we conclude that  $\psi^{-1}$ is an isomorphism.
This implies that  $\psi$ takes $V \otimes_{K_{0,G}, \F} K_{0,G}$ to $V$, which in turn shows  that $V$
is invariant under the Frobenius map.

In particular, for all nonnegative integers 
$s$, we have $\Delta_G(f)^{p^s}\in V$. 
Since $\dim_{K_{0,G}} V \leq N$, it follows that $\Delta_G(f), \Delta_G(f)^p, \ldots,
\Delta_G(f)^{p^N}$ must be linearly dependent over $K_{0,G}$.
Thus, there exist $c_0, c_1, \ldots, c_N \in K_{0,G}$, not all zero, such that
\[c_0 \cdot \Delta_G(f) +  c_1 \cdot \Delta_G(f)^p + \cdots + 
c_N \cdot \Delta_G(f)^{p^N} = 0 \,,\]
as desired.
\end{proof}

\subsection{Proof of \cref{thm:diag}} We will now proceed with the proof of the main result of this section.

\begin{proof}[Proof of \cref{thm:diag}] 
We apply \cref{thm:gendiag} with the group $G$ generated
by $(1, \ldots, 1)$. As already noted, the diagonal $\Delta_G$
coincides with the usual diagonal $\Delta$, under the identification 
$t \coloneqq t_1 \cdots t_n$.
Let $A(\bt, y)$ be the minimal polynomial of $f$, so that $A$ has degree $d$, total height $h$, and partial height $\bh=(h_1,\ldots,h_n)$.
Let $\pi_G$  and $C$  be the mapping and convex set defined
in the statement of \cref{thm:gendiag}. 
By \cref{thm:gendiag}, it remains to prove that $\Card \big(\pi_G(C \cap \NN^{n+1})\big) \leq N$, where $N$ is defined as in Equation~\eqref{eq: N31}. 

Consider an element $c \coloneqq (a_1, \ldots, a_n, b) \in C \cap \NN^{n+1}$. By definition of $C$, we have $-1 < b \leq d$, and since 
 $b$ is an integer, we conclude that $0 \leq b \leq d$. 
Moreover, by translating $c$ by an element of $G$, we may assume that $0 \leq a_i \leq h_i$ for
all $i \in \{1, \ldots, n\}$, and $\sum_{i=1}^n a_i\leq h$.
Define $a \coloneqq \min\{a_1, \ldots, a_n\}$ and, for each $i$, set $\tilde a_i \coloneqq a_i - a$.
Then  at least one of the first $n$ coordinates of $\tilde c \coloneqq (\tilde a_1, \ldots, \tilde a_n, b)$ 
is zero. 
Furthermore, we still have $0\leq \tilde a_i \leq h_i$ and $\sum_{i=1}^n \tilde a_i \leq h$. 
Since $\pi_G(c) = \pi_G(\tilde c)$, 
it follows that every element of  $\pi_G(C \cap \NN^{n+1})$ has a preimage in both of the following sets:  
\[
\mathcal E_1 \coloneqq \left\{ (a_1,\ldots,a_n,b) \in \mathbb N^{n+1} :  b\leq d,  \forall i, a_i 
\leq h_i, \exists i, a_i=0 \right\}\,,
\]
\[
\mathcal E_2 \coloneqq \left\{ (a_1,\ldots,a_n,b) \in \mathbb N^{n+1} : b\leq d, \sum_{i=1}^na_i \leq h, \exists i, a_i=0 \right\} \,.
\]
The cardinalities of these sets are given by
\[\Card(\mathcal E_1) = (d+1)\cdot \left(\prod_{i=1}^n (h_i+1) - \prod_{i=1}^nh_i\right) \,,\] 
\[\Card (\mathcal E_2) =   (d+1)\cdot \left( \binom{n+h}{n}-\binom{h}{n} \right) \, .\] 
Thus, we obtain that $ \Card \big(\pi_G(C \cap \NN^{n+1})\big) \leq \min \{\Card(\mathcal E_1) , \Card(\mathcal E_2)\} =   N$, as required.  
\end{proof}

\subsection*{Acknowledgements} 
The authors would like to thank the anonymous referees for their careful reading of an earlier version of this paper and for their valuable comments.

\bibliographystyle{alpha}

\end{document}